\documentclass[runningheads]{llncs}

\usepackage{bbding}  

\usepackage{amsmath}
\usepackage{amsfonts}
\spnewtheorem{domi_rule}[exercise]{Corollary}{\bfseries}{\itshape}
\spnewtheorem{condition}[problem]{Condition}{\bfseries}{\itshape}

\usepackage{graphicx}
\usepackage{float}
\usepackage{subfigure}
\usepackage{tkz-linknodes}
\usepackage{tikz,pgfplots}
\usetikzlibrary{arrows,shapes,intersections,decorations}
\pgfplotsset{compat=1.14}
\usepackage{algpseudocode}
\usepackage{algorithm}
\usepackage{algorithmicx}
\algrenewcommand\algorithmicrequire{\textbf{Input:}}   
\algrenewcommand\algorithmicensure{\textbf{Output:}}   

\usepackage{enumitem}

\usepackage{booktabs}
\usepackage{multirow}
\usepackage{tabularx}
\newcolumntype{Y}{>{\centering\arraybackslash}X}

\usepackage{hyperref}

\newcommand{\bigO}{\mathcal{O}}

\usepackage[linecolor=green!70!black, backgroundcolor=green!10, bordercolor=black, textsize=scriptsize, obeyFinal]{todonotes}

\begin{document}

\title{A Fast Exact Algorithm for Airplane Refueling Problem \thanks{This work is supported by Key Laboratory of Management, Decision and Information Systems, CAS.}}
%
%
\author{Jianshu Li \inst{1, 2 (}\Envelope\inst{)}  \and  
Xiaoyin Hu\inst{1, 2} \and
Junjie Luo\inst{1, 2} \and
Jinchuan Cui \inst{1}}
\authorrunning{J. Li et al.}
%
\institute{Academy of Mathematics and Systems Science Chinese Academy of Sciences, Beijing 100190, China\\
\email{\{ljs, hxy, luojunjie, cjc\}@amss.ac.cn} \and
School of Mathematical Sciences, University of Chinese Academy of Sciences, Beijing 100049, China}
\maketitle              
\begin{abstract}
We consider the airplane refueling problem,
where we have a fleet of airplanes that can refuel each other.
Each airplane is characterized by specific fuel tank volume and fuel consumption rate,
and the goal is to find a drop out order of airplanes 
that last airplane in the air can reach as far as possible.
This problem is equivalent to the scheduling problem $1||\sum w_j (- \frac{1}{C_j})$.
Based on the dominance properties among jobs,
we reveal some structural properties of the problem and
propose a recursive algorithm to solve the problem exactly.
The running time of our algorithm is directly related to the number of schedules that do not violate the dominance properties.
An experimental study shows our algorithm outperforms state of the art exact algorithms and
is efficient on larger instances.
\end{abstract}

\section{Introduction} 
\label{sec:introduction}

The \emph{airplane refueling problem}, 
originally introduced by Gamow and Stern \cite{puzzle_math_1958},
is a special case of single machine scheduling problem.
Consider a fleet of several airplanes with certain fuel tank volume and fuel consumption rate. 
Suppose all airplanes travel at the same speed and they can refuel each other in the air.
An airplane will drop out of the fleet if it has given out all its fuel to other airplanes.
The goal is to find a drop out order so that the last airplane can reach as far as possible.

\subsubsection{Problem definition} 
\label{ssub:problem_definition}
In the original description of airplane refueling problem, 
all airplanes are defaulted to be identical.
Woeginger \cite{chrobaketal:DSP:2010:2536} generalized this problem that 
each airplane $j$ can have arbitrary tank volume $w_j$ and consumption rate $p_j$.
Denote the set of all airplanes by $J$,
a solution is a drop out order $\sigma: \{1, 2, \dots, |J|\} \mapsto J$, 
where $\sigma(i) = j$ if airplane $j$ is the $i^{\text{\tiny th}}$ airplane to leave the fleet.
As a result, the objective value of the drop out order $\sigma$ is:

\[
	\sum_{j = 1}^{n} \Big( {w_{\sigma(j)}} 
	\text{\bfseries \large{/}} \text{\tiny \ }
	{\sum_{k = j}^{n}p_{\sigma(k)}} \Big).
\]

As pointed out by Vásquez \cite{vasquez_for_2015}, we can rephrase the problem as a single machine scheduling problem, which is equivalent to finding a permutation $\pi$(the reverse of $\sigma$) that minimizes:

\[
	\sum_{j = 1}^{n} \Big( - {w_{\pi(j)}} 
	\text{\bfseries \large{/}} \text{\tiny \ }
	{\sum_{k = j}^{n}p_{\pi(k)}} \Big) = \sum_{j=1}^{n} -w_j / C_j,
\]
where $C_j$ is the completion time of job $j$,
and $p_j, w_j$ correspond to its processing time and weight, respectively.
This scheduling problem is specified as $1||\sum w_j (- \frac{1}{C_j})$ using the classification scheme introduced by Graham et al. \cite{graham_optimization_1979}.

In the rest of this paper, we study the equivalent scheduling problem instead.



\subsubsection{Dominance properties} 
\label{ssub:dominance_relation}
Since the computational complexity status of airplane scheduling problem is still open \cite{chrobaketal:DSP:2010:2536},
existing algorithms find the optimal solution with branch and bound method.
While making branching decisions in a branch and bound search, 
it would be much more useful if we know the dominance relation among jobs.
For example,
if we know job $i$ always precedes job $j$ in an optimal solution,
we can speed up the searching process by pruning all the branches with job $i$ processed after job $j$.
Let the start time of job $j$ be $t_j$, 
we refer to \cite{durr_order_2014} for the definition of \emph{local dominance} and \emph{global dominance} :

\begin{itemize}
	\item \textit{local dominance:}
	Suppose job $j$ starts at time $t$ and is followed directly by job $i$ in a schedule. 
	If exchanging the positions of jobs $i,j$ strictly improves the cost,
	we say that \emph{job $i$ locally dominants job $j$ at time $t$} and denote this property by $i \prec_{l(t)} j$.
	If $i \prec_{l(t)} j$ holds for all $t \in [a, b]$, we denote it by $i \prec_{l[a, b]} j$.

	\item \textit{global dominance:}
	Suppose schedule $S$ satisfies $a \leq t_j \leq t_i - p_j \leq b$.
	If exchanging the positions of jobs $i,j$ strictly improves the cost,
	we say that \emph{job $i$  globally dominants job $i$ in time interval $[a, b]$} 
	and denote this property by $i \prec_{g[a, b]} j$.
	If it holds that $i \prec_{g[0, \infty)} j$, we denote this property by $i \prec_{g} j$.
\end{itemize}

We call the schedule that do not violate the dominance properties as \emph{potential schedule}.
The effect of dominance properties is to narrow the search space to the set of all potential schedules,
whose cardinality is much smaller than $n!$, 
the number of all job permutations.


\subsubsection{Related work} 
\label{ssub:related_work}
Airplane refueling problem is a special case of a more general scheduling problem $1||\sum w_j C_j^\beta$.
For most problems of this form, 
including airplane refueling problem,
no polynomial algorithm has been found.
Existing methods resort to approximation algorithms or branch and bound schemes.

For approximations,
several constant factor approximations and polynomial time approximation scheme(PTAS) 
have been devised for different cost functions \cite{bansal_geometry_2014,cheung_primal-dual_2017,hohn_how_2018,megow_dual_2013}.
Recently, Gamzu and Segev\cite{gamzu_polynomial-time_2019} gave the first polynomial-time approximation scheme for airplane refueling problem.

The focus of exact methods is to find stronger dominance properties.
Following a long series of improvements \cite{bagga_node_1980,croce_minimizing_1993,bader_experimental_2012,mondal_improved_2000-1,sen_minimizing_1990},
Dürr and Vásquez \cite{durr_order_2014} conjectured that 
for all cost functions of the form $f_j(t) = w_jt^\beta, \beta >0$ and all jobs $i, j$, 
$i \prec_l j$ implies $i \prec_g j$.
Latter, Bansal et al. \cite{bansal_localglobal_2017} confirmed this conjecture,
and they also gave a counter example of the generalized conjecture that 
$i \prec_{l[a, b]} j$ implies $i \prec_{g[a, b]} j$.
For airplane refueling problem, 
Vásquez \cite{vasquez_for_2015} proved that $i \prec_{l[a, b]} j$ implies $i \prec_{g[a, b]} j$.
The establish dominance properties are commonly incorporated into a branch and bound scheme, such as Algorithm $A^*$ \cite{hart_formal_1968}, to speed up the searching process.

\subsubsection{Our contribution} 
\label{ssub:our_contribution}
Existing branch and bound algorithms search for potential schedules in a trail and error manner.
Specifically, when making branching decisions,
it is unknown whether current branch contains any potential schedule
unless we exhaust the entire branch.
So if we can prune all the branches that do not contain potential schedule,
it will considerably improve the efficiency of the searching process.

In this paper we give an exact algorithm with the merit above for airplane refueling problem.
Specifically, 
every branch explored by our algorithm is guaranteed to contain at least one potential schedule,
and the time to find each potential schedule is bounded by $\bigO(n^2)$.
Numerical experiments exhibit empirical advantage of our algorithm over the $A^*$ algorithms proposed by former studies,  
and the advantage is more significant on instances with more jobs.

The main difference between previous methods and our algorithm is that
instead of branching on possible succeeding jobs,
we branch on the possible start times of a certain job.
To this end,
we introduce a prefixing technique to determine the possible start times of a certain job in a potential schedule.
In addition, 
the relative order of other jobs regarding that certain job is also decided.
Thus, each branch divides the original problem into subproblems with fewer jobs
and we can solve the problem recursively.

\subsubsection{Organization} 
\label{ssub:organization}
The rest of the paper is organized as follows.
In section 2,  we introduce a new auxiliary function and give a concise representation of dominance property.
Section 3 establishes some useful lemmas.
We present our algorithm in section 4 and experimental results in section 5.
Finally, section 6 concludes the paper.

%
%

\section{Preliminaries} 
\label{sec:preliminaries}

A \emph{dominance rule} stipulates the local and global dominance properties among jobs.
We can refer to a dominance rule by the set $R:= \cup_{i, j \in J} \{r_{ij}\}$,
where $r_{ij}$ represents the dominance properties between jobs $i, j$ specified by the rule.
Given a dominance rule $R$, we formally define the potential schedule as follows.

\begin{definition}[Potential Schedule]
	We call $S$ a potential schedule with respect to dominance rule R
	if for all $r_{ij} \in R$,
	start times of jobs $i,j$ in $S$ do not violate relation $r_{ij}$.
\end{definition}


\subsubsection{Auxiliary function} 
\label{ssub:auxiliary_function}
For each job $j$, we introduce an auxiliary function $\varphi_j(t)$,
where $t \geq 0$ represents the possible start time:
\[
\varphi_j(t) = \frac{w_j}{p_j(p_j + t)}.
\]

We remark that $\varphi_j(t)$ is well defined since the processing time $p_j$ is positive.
With the help of this auxiliary function, 
we can obtain a concise characterization of local precedence between two jobs.

\begin{lemma}
\label{lemma: local_new}
	For any two jobs $i, j$ and time $t \geq 0$, $i \prec_{l(t)} j$ if and only if $\varphi_{i}(t) > \varphi_{j}(t)$.
\end{lemma}

\begin{proof}
	Suppose $S = AijB$ is an arbitrary schedule, where job $i$ starts at time $t$ and job $j$ is processed directly after job $i$.
	We use $F(S)$ to represent the cost of the schedule $S$.
	To explore the local dominance relation between job $i$ and job $j$, 
	we need to check the impact of swapping the order of $i, j$,
	which leads to:
	\begin{align*}
		F(AijB) - F(AjiB) & = \frac{w_i}{p_i + t} + \frac{w_j}{p_j + p_i + t} 
		                      - \frac{w_j}{p_j + t} - \frac{w_i}{p_i + p_j + t} \\
						  & = w_i \cdot \frac{p_j}{(p_i + t)(p_i + p_j + t)} -    w_j \cdot \frac{p_i}{(p_j + t)(p_i + p_j + t)} \\
						  & = \frac{p_ip_j}{p_i + p_j + t}(\frac{w_i}{p_i(p_i + t)} - \frac{w_j}{p_j(p_j + t)}) \\
						  & = \frac{p_ip_j}{p_i + p_j + t}(\varphi_i(t) - \varphi_j(t)).
	\end{align*}
	Since $\frac{p_ip_j}{p_i + p_j + t} > 0$, the effect of exchanging jobs $i, j$ depends on the term $(\varphi_i(t) - \varphi_j(t))$,
	which completes the proof.
	\qed
\end{proof}

%

\subsubsection{Dominance rule} 
\label{ssub:modified_rule}

Vásquez \cite{vasquez_for_2015} proved the following dominance properties of airplane refueling problem.

\begin{theorem}[Vásquez]
	\label{theorem: vas}
	For all jobs $i, j$ and time points $a, b$, the dominance property $i \prec_{l[a, b]} j$ implies $i \prec_{g[a, b]} j$.
\end{theorem}

Based on Theorem \ref{theorem: vas} and Lemma \ref{lemma: local_new} we obtain the following dominance rule of airplane refueling problem.

\begin{corollary}
\label{rule: new_rule}
	For any two jobs $i, j$ with $w_i > w_j$:

	\begin{enumerate}
		\item If $\varphi_i(t) \geq \varphi_j(t)$ for $t \in [0, \infty)$, $i \prec_g j$;
		\item If $\varphi_j(t) \geq \varphi_i(t)$ for $t \in [0, \infty)$, $j \prec_g i$;
		\item else $\exists$ $t^*_{ij} = \frac{w_jp_i^2 - w_ip_j^2}{w_i p_j - w_j p_i} > 0$ with:
		\begin{itemize}
			\item $\varphi_i(t) > \varphi_j(t)$ for $t \in [0, t^*_{ij})$, $i \prec_{g[0, t^*_{ij})} j$;
			\item $\varphi_j(t) \geq \varphi_i(t)$ for $t \in [t^*_{ij}, \infty)$, $j \prec_{g[t^*_{ij}, \infty)}i$.
		\end{itemize}
	\end{enumerate}

\end{corollary}

Actually, this rule is equivalent to the rule given in \cite{vasquez_for_2015},
whereas we use $\varphi$ 
to indicate the dominance relation between any jobs $i, j$.
Figure \ref{fig: auxilary_function} gives an illustrative example of the scenario where 
$t^*_{ij} = \frac{w_jp_i^2 - w_ip_j^2}{w_i p_j - w_j p_i} > 0$.

In this paper, we care about potential schedules with respect to Corollary \ref{rule: new_rule}.
We refer to these schedules as \emph{potential schedules} for short.

\begin{figure}[H]
\centering
	\begin{tikzpicture}[domain=-2:4,yscale=1,samples=200,>=latex,thick]
	  \draw[thick,->] (-0.5,0.5) -- (4,0.5) node[right] {$t$};
	  \draw[thick,->] (0,0) -- (0,4) node[below left] {$\varphi(t)$};
	  \draw (0,0.5) node[below left] {O};
	  \coordinate (O) at (0,0);

	  \draw[name path=i, domain=0:4, color=black] plot (\x, {(12)/(2*(2+ \x)}) node[right] {$y=\varphi_i(t)$};
	  \draw[name path=j, domain=0:4, color=black] plot (\x, {(162)/(9*(9 + \x)}) node[right] {$y=\varphi_j(t)$};

	  \draw [name intersections={of=i and j, by=x}, dashed, thin] (x) -- (1.5, 0.5) node[below] {$t^*_{ij}$};
	\end{tikzpicture}
\caption{Illustration of dominance property between job $i, j$. For $t \in [0, t^*_{ij})$, $\varphi_i(t) > \varphi_j(t)$, $i \prec_{g[0, t^*_{ij})} j$;
for $t \in [t^*_{ij}, \infty)$, $\varphi_i(t) \leq \varphi_j(t)$, $j \prec_{g[t^*_{ij}, \infty)}i$.}
\label{fig: auxilary_function}
\end{figure}
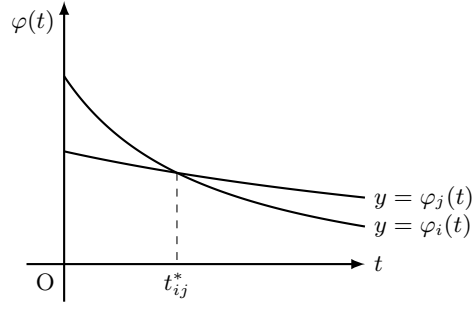


%
%

\section{Technical lemmas} 
\label{sec:technical_lemmas}

In this section, we establish several important properties concerning potential schedules.

\subsection{Relative order between two jobs} 
\label{sub:relative_order_between_two_jobs}
To begin with, 
we show that the dominance rule makes it impossible for a job to start within some time intervals in a potential schedule.
Let $T:= \sum_{j \in J} p_j$ be the total processing time, 
we have:

\begin{lemma}
\label{lemma: banned_interval}
	For two jobs $i, j \in J$ with $\varphi_i(0) > \varphi_j(0)$ and $t^*_{ij} \in (0, T)$,
	if in a complete schedule $S$ job $i$ starts in time interval $[t^*_{ij}, t^*_{ij} + p_j)$,
	then $S$ is not a potential schedule.
\end{lemma}

\begin{proof}
	By Corollary \ref{rule: new_rule} it follows that $i \prec_{g[0, t^*_{ij})} j$ 
	and $j \prec_{g[t^*_{ij},\infty)} i$.

	If job $j$ precedes job $i$ in $S$, then $t_j < t_j + p_j \leq t_i$.
	Since $t_i \in [t^*_{ij}, t^*_{ij} + p_j)$,
	we have:

	\[
		0 \leq t_j \leq t_i - p_j < t^*_{ij},
	\]

	which leads to a violation of the relation $i \prec_{g[0, t^*_{ij})} j$ .

	Otherwise job $j$ is processed after job $i$ in $S$. 
	In a similar manner we can derive:

	\[
		t^*_{ij} \leq t_i \leq t_j - p_i < \infty.
	\]

	Again relation $j \prec_{g[t^*_{ij}, \infty)} i$ is not satisfied.
	\qed
\end{proof}

According to Lemma \ref{lemma: banned_interval}, 
job $i$ starts either in  $[0, t^*_{ij})$ or in $[t^*_{ij} + p_j, T - p_i)$ in a potential schedule. 
Next lemma shows that if we fix the start time of job $i$ to one of these intervals, 
the relative order of job $i, j$ in a potential schedule is already decided.

\begin{lemma}
\label{lemma: left_or_right}
	For two jobs $i, j \in J$ with $\varphi_i(0) > \varphi_j(0)$ and $t^*_{ij} \in (0, T)$, 
	suppose $S$ is a schedule with $t_j \not \in [t^*_{ij}, t^*_{ij} + p_j)$.

	\begin{enumerate}
		\item If $t_i < t^*_{ij}$, 
		job $i$ should be processed before job $j$,
		otherwise $S$ is not a potential schedule;
		\item If $t_i \geq t^*_{ij} + p_j$, 
		job $i$ should be processed after job $j$,
		otherwise $S$ is not a potential schedule.
	\end{enumerate}

\end{lemma}

\begin{proof}
	According to Corollary \ref{rule: new_rule} we have
	$i \prec_{g[0, t^*_{ij})} j$ and $j \prec_{g[t^*_{ij}, \infty)} i$.

	For the first scenario,
	processing job $j$ before job $i$ in $S$ would imply $t_j + p_j < t_i$, 
	which leads to the following inequalities:

	\[
		0 \leq t_j \leq t_i - p_j < t^*_{ij}.
	\]

	Thus we have reached a negation to $i \prec_{g[0, t^*_{ij})} j$.

	Similarly, for the second scenario if job $i$ precedes job $j$ we will have:

	\[
		t^*_{ij} \leq t_i \leq t_j - p_i < T.
	\]

	Relation $j \prec_{g[t^*_{ij}, \infty)} i$ does not hold,
	which concludes the proof.
	\qed
\end{proof}

Conditioning on the start time of job $i$,
Lemma \ref{lemma: banned_interval} and Lemma \ref{lemma: left_or_right} provide a complete characterization of 
the relative order between jobs $i, j$ in a potential schedule.
See Figure \ref{fig: relative_order} for an overview of these relations.

\begin{figure}
\centering
	\begin{tikzpicture}[domain=-2:8,yscale=1,samples=200,>=latex,thick]
	  \draw[thick,-] (0,0) -- (8,0);

	  \draw (0, 0.05) -- (0, -0.05) node[below] {$0$};
	  \draw (3, 0.05) -- (3, -0.05) node[below] {$t^*_{ij}$};
	  \draw (5, 0.05) -- (5, -0.05) node[below] {$t^*_{ij} + p_j$};
	  \draw (8, 0.05) -- (8, -0.05) node[below] {$T$};

	  \draw[thick, decoration={brace, amplitude=5}, decorate] (0,0.1) -- (2.95,0.1);
	  \draw[thick, decoration={brace, amplitude=5}, decorate] (3.05,0.1) -- (4.95,0.1);
	  \draw[thick, decoration={brace, amplitude=5}, decorate] (5.05,0.1) -- (8,0.1);

	  \draw (1.5, 0.2) node[above] {$i$ before $j$};
	  \draw (4, 0.2) node[above] {banned interval};
	  \draw (6.5, 0.2) node[above] {$j$ before $i$}; 
	\end{tikzpicture}
\caption{The relative order of job $i$ and job $j$ in a potential schedule conditioned on the start time of job $i$ with $t^*_{ij} \in (0, T)$ and $\varphi_i(0) > \varphi_j(0)$.
When $t_i \in [0, t^*_{ij})$, job $i$ should precede job $j$;
when $t_i \in [t^*_{ij}, t^*_{ij} + p_j)$, the resulting schedule will violate dominance property;
when $t_i \in [t^*_{ij} + p_j, T)$, job $j$ should precede job $i$.}
\label{fig: relative_order}
\end{figure}
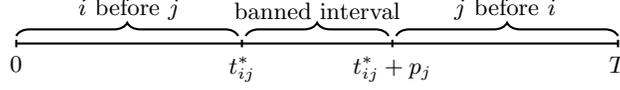

%

\subsection{Possible start times in a potential schedule} 
\label{sub:possible_start_times_in_a_potential_schedule}

In this subsection we consider a general \emph{sub-problem} scenario.
Let $J' \subseteq J$ be a set of jobs 
to be processed consecutively at start time $t_o$, where $t_o \in [0, T - \sum_{j \in J'} p_j]$.
For convenience we denote the completion time $t_o + \sum_{j \in J'} p_j$ by $t_e$.
Notice that when $J' = J$, we will have $t_o = 0$ by definition, the setting above describes the original problem.

We further denote the partial schedule of $J'$ by $S'$. 
If $S'$ does not violate any dominance rule,
we call it as a \emph{partial potential schedule}.

Now consider job $\alpha \in J'$ such that: 

\[
\alpha = \operatorname*{arg\,max}_i \varphi_i(t_o).
\]

When there are more than one job with the maximum $\varphi(t_o)$,
the tie breaking rule is to choose the job with the maximum $\varphi(t_e)$.

Our goal is to study the possible start times of job $\alpha$ in a partial potential schedule $S'$.
We start with analyzing 
the precedence relation of job $\alpha$ with other jobs in a partial potential schedule.
To that end, 
we divide time interval $[t_o, t_e]$ into consecutive subintervals 
with the time set:

\[
	C_{J'}^{\alpha} := \{c_{\alpha j} | c_{\alpha j} = t^*_{\alpha j} + p_j, c_{\alpha j} \in (t_o, t_e), j \in J'\setminus \{\alpha\} \} \cup \{t_o, t_e\}.
\]

We re-index the set $C_{J'}^{\alpha}$ according to their value rank,
that is, if a time point ranks the $q^{\text{\tiny th}}$ in the set, 
it will be denoted by $c_q$.
Besides, we use mapping $M_{J'}^{\alpha}: J' \mapsto \mathbb{Z^+}$ to maintain job information of the index.
The mapping is defined as: 

\[
	M_{J'}^{\alpha}(j) = \begin{cases}
		1, 													  & \text{if $j = \alpha$}; \\
		\text{the rank of $c_{\alpha j}$ in $C_{J'}^\alpha$}, & \text{if $t^*_{\alpha j} \in (t_o, t_e)$}; \\
		|J'| + 1,                                             & \text{if $t^*_{\alpha j} \notin (t_o, t_e)$}.

	\end{cases}
\]

We consider the start time of job $\alpha$ in each subinterval $[c_q, c_{q + 1})$.
Next lemma shows that once we fix $t_\alpha$ to a subinterval,
the positions of all remaining jobs in a potential schedule relative to $\alpha$ are determined.

\begin{lemma}
\label{lemma: J_l_J_r}
	Suppose $S'$ is partial potential schedule that job $\alpha$ starts within time interval $[c_q, c_{q + 1})$,
	then all $j$ with $M_{J'}^{\alpha}(j) \leq q$ should proceed $\alpha$, 
	and the rest jobs with $M_{J'}^{\alpha}(j) > q$ should come after $\alpha$.
\end{lemma}

\begin{proof}
	Suppose $j$ is an arbitrary job in the set $J'$ other than $\alpha$,
	by the definition of job $\alpha$ and Corollary \ref{rule: new_rule} we know either it is dominated by job $\alpha$ in time interval $[t_o, t_e]$,
	or there exists $t^*_{ij} \in (t_o, t_e)$ that 
	$\alpha \prec_{g[t_o, t^*_{ij})} j$ and $j \prec_{g[t^*_{ij}, t_e)} \alpha$.

	For the first case, we always have $M_{J'}^{\alpha}(j) = |J'| + 1 > q$, 
	job $j$ comes after job $\alpha$.

	While for the second case, 
	Lemma \ref{lemma: banned_interval} and Lemma \ref{lemma: left_or_right} apply, we have:
	\begin{itemize}
		\item If $M_{J'}^{\alpha}(j) \leq q$, 
			  it implies that $t_\alpha \geq t^*_{\alpha j} + p_j$. 
			  By Lemma \ref{lemma: left_or_right} job $j$ should precedes job $\alpha$ in a potential schedule.
		\item Otherwise $M_{J'}^{\alpha}(j) < q$, 
		      it follows that $t_\alpha < t^*_{\alpha j} + p_j$.
	          Since $S'$ is a partial potential schedule, 
	          Lemma \ref{lemma: banned_interval} rules out the possibility that $t_\alpha \in [t^*_{\alpha j}, t^*_{\alpha j} + p_j)$, 
	          we have $t_\alpha < t^*_{\alpha j}$.
	          Again by Lemma \ref{lemma: left_or_right} job $j$ should come after job $\alpha$.
	\end{itemize}
	\qed
\end{proof}

At last,
for the possible start times of job $\alpha$ in a partial potential schedule, we have:

\begin{lemma}
\label{lemma: at_most_start_times}
For every time interval $[c_q, c_{q + 1})$, there is at most one possible start time for job $\alpha$ in $[c_q, c_{q + 1})$. Thus, there are at most $|J'|$ possible start times of job $\alpha$ in a partial potential schedule $S'$.
\end{lemma}

\begin{proof}

	Suppose that $t_{\alpha} \in [c_q, c_{q + 1})$, according to Lemma \ref{lemma: left_or_right}, the positions of all remaining jobs in $J' \setminus \{\alpha\}$ relative to $\alpha$ are determined.
	We denote the set of jobs before and after job $\alpha$ as $J_l, J_r$, respectively.
	The start time of job $\alpha$ is determined by $J_l$. 
	More precisely, we have: $t_\alpha = t_o + \sum_{j \in J_l} p_j.$


	While if $t_o + \sum_{j \in J_l} p_j \notin [c_q, c_{q + 1})$,
	we have reached a conflict,
	which means there is no partial potential schedule that
	job $\alpha$ starts in current subinterval.
	As a result, there is at most one possible start time for job $\alpha$ in each subinterval.
	\qed 
\end{proof}


%
%

\section{Algorithm} 
\label{sec:algorithm}

In this section we devise an exact algorithm for airplane refueling problem 
and analyze its running time.

\subsection{An exact algorithm for air plane refueling problem} 
\label{sub:an_exact_algorithm_for_air_plane_refueling_problem}

We start with some notations.
For each job pair $i, j$ of Lemma \ref{lemma: banned_interval},
we name the time interval $b_{ij} = [t^*_{ij}, t^*_{ij} + p_j)$ as the \emph{banned interval} of job $i$ imposed by job $j$,
and denote $B_i := \cup_{j \in J \setminus {i}} b_{ij}$ as the union of all banned intervals of job $i$.

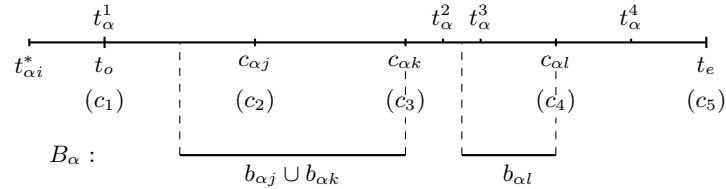
\begin{figure}[H]
\centering
	\begin{tikzpicture}[domain=-2:8,yscale=1,samples=200,>=latex,thick]
	  \draw[thick,-] (-1,0) -- (8,0); 

	  \draw (-1, 0.05) -- (-1, -0.05) node[below] {$t^*_{\alpha i}$};
	  \draw (0, 0.075) -- (0, -0.075) node[below] {$t_o$};
	  	\draw (0, -0.5) node [below] {$(c_1)$};
	  \draw (2, 0.05) -- (2, -0.05) node[below] {$c_{\alpha j}$};
	  	\draw (2, -0.5) node [below] {$(c_2)$};	  
	  \draw (4, 0.05) -- (4, -0.05) node[below] {$c_{\alpha k}$};
	  	\draw (4, -0.5) node [below] {$(c_3)$};	
	  \draw (6, 0.05) -- (6, -0.05) node[below] {$c_{\alpha l}$};
	  	\draw (6, -0.5) node [below] {$(c_4)$};	
	  \draw (8, 0.075) -- (8, -0.075) node[below] {$t_e$};
	  	\draw (8, -0.5) node [below] {$(c_5)$};	

	  \draw (0, 0.05) node [above] {$t^1_{\alpha}$};
	  \draw (4.5, 0) -- (4.5, 0.05) node [above] {$t^2_{\alpha}$}; 
	  \draw (5, 0) -- (5, 0.05) node [above] {$t^3_{\alpha}$};  
	  \draw (7, 0) -- (7, 0.05) node [above] {$t^4_{\alpha}$};

	  \draw (0, -1.5) node [left] {$B_\alpha :$};

	  \draw[thick, -] (1,-1.5) -- (4,-1.5); 
		  \draw[dashed, thin] (1, -1.5) -- (1, 0);
		  \draw[dashed, thin] (4, -1.5) -- (4, -1);
		  \draw[dashed, thin] (4, -0.7) -- (4, -0.35);

	  \draw[thick, -] (4.75, -1.5) -- (6, -1.5);
		  \draw[dashed, thin] (4.75, -1.5) -- (4.75, 0);
		  \draw[dashed, thin] (6, -1.5) -- (6, -1);
		  \draw[dashed, thin] (6, -0.7) -- (6, -0.35);	

	  \draw (2.5, -1.5)  node [below] {$b_{\alpha j} \cup b_{\alpha k}$};
	  \draw (5.5, -1.5)  node [below] {$b_{\alpha l}$};
	\end{tikzpicture}
\caption{Possible start times of job $\alpha$.}
\label{fig: pos_start_time}
\end{figure}

Given an instance of airplane refueling problem, 
we find the job $\alpha$ and try to start $\alpha$ in each interval induced by $C_{J'}^{\alpha}$.
For an interval $[c_q, c_{q+1})$ and corresponding $J_l, J_r$,
the possible start time of $\alpha$ will be $t_\alpha = \sum_{j \in J_l} p_j$.
If $t_\alpha \in [c_q, c_{q+1})$ and $t_\alpha \notin B_\alpha$, 
then we have found a potential start time of job $\alpha$.

Figure \ref{fig: pos_start_time} provides an illustrative example of the above procedure, 
where $J' = \{i, j, k, l, \alpha \}$.
As shown in the figure, 
the set $C^{\alpha}_{J'}$ divided time interval $[t_o, t_e)$ into 4 subintervals.
Job $i$ should always behind job $\alpha$ since $t^*_{\alpha i} < t_o$.
For other jobs in $J'$, 
condition on the start time of job $\alpha$ we have:

\begin{enumerate}
	\item $t_\alpha \in [c_1, c_2)$: Job $\alpha$ serves as the first job and $t_\alpha^1 = t_o$.
	\item $t_\alpha \in [c_2, c_3)$: Job $j$ precedes job $\alpha$ in a potential schedule.
	As $t_\alpha^2 = t_o + p_j > c_3$, job $\alpha$ can not start in this interval.
	\item $t_\alpha \in [c_3, c_4)$: Job $j, k$ precedes job $\alpha$ in a potential schedule.
	As $t_\alpha^3 = t_o + p_j + p_k$ is in banned intervals of job $\alpha$, 
	job $\alpha$ can not start in this interval.
	\item $t_\alpha \in [c_4, c_5)$: Job $j, k, l$ precedes job $\alpha$,
	$t_\alpha^4 = t_o + p_j + p_k + p_l$.
\end{enumerate}

Once we find a possible start time of job $\alpha$, 
we can solve $J_l$ and $J_r$ recursively with the procedure above 
until the subproblem has only one job.
See Algorithm \ref{algo: fast_schedule} for the pseudo code of our algorithm,
where the initial inputs are the complete job set $J$ and original start time $t=0$.

\begin{algorithm}[H]
	\caption{Fast Schedule}
	\label{algo: fast_schedule}
	\begin{algorithmic}[1]
	\Require{Job set $J$ and start time $t$.}
	\Ensure{The optimal schedule and the optimal cost.}
	\Function{FastSchedule}{$J, t$}
		\If{$J$ is empty}
			\State \Return{$0$, $[]$}
		\ElsIf{$J$ contains only one job $j$}
			\State \Return{$\frac{w_j}{p_j + t}$, $[j]$}
		\EndIf

		\State Find job $\alpha$ 


		\State $J_l, J_r, opt \gets [], J, 0$

		\For{$c_{q} \in C_J^\alpha$}
			\State Find job $i$ with $M_{J}^{\alpha}(i) = q$ \Comment{$i = \alpha$ for $q=1$, i.e. job $\alpha$ is the first job}
			\State $J_l \gets J_l \cup \{i\}$ 
			\State $J_r \gets J_r \setminus \{i\}$


			\If{$c_{q} \leq t + \sum_{i \in J_l} p_i  <{c}_{q+1} \And \ (t + \sum_{i \in J_l} p_i) \notin B_\alpha$} \label{line: new_branch} \Comment{new branch}

				\State $opt_l, seq_l \gets $ \Call{FastSchedule}{$J_l \setminus \{\alpha\}, t$}
				\State $opt_r, seq_r \gets $ \Call{FastSchedule}{$J_r, t + \sum_{i \in J_l} p_i + p_\alpha$}

				\If{$opt_l + opt_r + \frac{w_\alpha}{p_\alpha + \sum_{i \in J_l} p_i} > opt$}
					\State $opt, seq \gets (opt_l + opt_r + \frac{w_\alpha}{p_\alpha + \sum_{i \in J_l} p_i})$, $[seq_l, \alpha, seq_r]$  
				\EndIf
			\EndIf
		\EndFor

		\State \Return{opt, seq}
	\EndFunction
	\end{algorithmic}
\end{algorithm}

Next, we prove the correctness of Algorithm \ref{algo: fast_schedule}.

\begin{theorem}
	Algorithm \ref{algo: fast_schedule} returns the optimal solution of airplane refueling problem.
\end{theorem}

\begin{proof}
	We only need to show that Algorithm \ref{algo: fast_schedule} does not eliminate any potential schedules.

	While locating possible start times of job $\alpha$,
	our algorithm drops out schedules with $t_\alpha$ in banned intervals or $t_\alpha$ exceeds current interval.
	According to Lemma \ref{lemma: banned_interval} and Lemma \ref{lemma: at_most_start_times}, 
	all the excluded schedules violate Corollary \ref{rule: new_rule}.

	Each recursive call also drops out schedules that have jobs in $J_l$ processed after job $\alpha$
	or jobs in $J_r$ processed before job $\alpha$.
	By Lemma \ref{lemma: left_or_right},
	these schedules are not potential schedules either.
	\qed
\end{proof}
%

\subsection{Running time of Algorithm \ref{algo: fast_schedule}} 
\label{sub:running_time_of_algorithm_algo: fast_schedule}

In this section we analyze the running time of Algorithm \ref{algo: fast_schedule} with respect to the number of potential schedules.
Our main result can be stated as follows:

\begin{theorem}
\label{theorem: num_sol}
	Algorithm \ref{algo: fast_schedule} finds the optimal solution of airplane refueling problem in 
	$\bigO(n^2(\log n + K))$ time, 
	where $K$ is the number of potential schedules with respect to Corollary \ref{rule: new_rule}.
\end{theorem}

Before proving this theorem, we need to establish some properties of Algorithm \ref{algo: fast_schedule}.
We start by showing that for any job set $J'$ that starts at time $t$, 
there is at least one partial potential schedule.

\begin{lemma}
\label{lemma: greedy}
	Suppose $J' \subseteq J$ is a set of job that starts at time $t$,
	then there exists a procedure that can find one partial potential schedule for $J'$.
\end{lemma}

\begin{proof}
	Consider the following procedure that construct the schedule successively.
	At each stage,
	choose the job with the maximum $\varphi(C)$ from the unscheduled jobs as the next job,
	where $C$ is the total processing time of the scheduled jobs plus $t$.
	We claim this procedure returns a partial potential schedule of $J'$.

	Let $S'$ be the schedule returned by the procedure above.
	Assume by contradiction that there are two jobs $i,j $ that violate dominance properties.
	Suppose $i$ precedes $j$ in $S'$, 
	then we have $j \prec_{g[t_i, t_j - p_i)} i$, 
	which implies $\varphi_j(t_i) > \varphi_i(t_i)$ by Lemma \ref{lemma: local_new}.
	In that case, we should have chosen job $j$ at time $t_i$ instead of $i$, 
	absurdity is obtained.
	\qed
\end{proof}


While Lemma \ref{lemma: greedy} concerns about a single job set,
the following lemma describes the relationship between two job sets $J_l$ and $J_r$.

\begin{lemma}
\label{lemma: independent}
	Given start time $t_{\alpha} \in [c_{q}, c_{q + 1})$ of job $\alpha$ 
	with $t_\alpha \in [t_o, t_e-p_\alpha]$ and $t_\alpha \notin B_\alpha$, 
	for any two jobs $j \in J_l$ and $k \in J_r$ with $M_{J'}^{\alpha}(j) \leq q < M_{J'}^{\alpha}(k)$,
	if job $j$ is scheduled before job $\alpha$ and job $k$ is scheduled after job $\alpha$,
	then no matter where job $j$ and job $k$ start, 
	dominance properties among $\alpha,j$ and $k$ will not be violated.
\end{lemma}

\begin{proof}
	The dominance relations between jobs $\alpha,j$ and jobs $\alpha,k$ are satisfied, 
	we only need to consider the relation between job $j$ and job $k$.
	From the definition of $c_q$, it follows that 

	\[
		t^*_{\alpha j} + p_j \leq t_{\alpha} < t^*_{\alpha k} + p_k.
	\]

	Since $t_\alpha$ is not in banned interval $b_{\alpha k} = [t^*_{\alpha k}, t^*_{\alpha k} + p_k)$, 
	the equation above leads to:

	\[
		t^*_{\alpha j} < t_\alpha < t^*_{\alpha k}.
	\]

	We distinguish the cases whether job $j$ and job $k$ have global precedence in time interval $[t_o, t_e]$.

	\begin{case}[Global dominance]\\
	Job $j$ and job $k$ have global precedence in time interval $[t_o, t_e]$.
		\begin{enumerate}
			\item $j \prec_{g[t_o, t_e]} k$: Since job $j$ is scheduled before job $k$, precedence rule is satisfied.
			\item $k \prec_{g[t_o, t_e]} j$: We will show this scenario is impossible. 
			At time $t_\alpha$, since $t^*_{\alpha j} < t_\alpha < t^*_{\alpha k}$ we have $j \prec_{l(t_\alpha)} \alpha$ and $\alpha \prec_{l(t_\alpha)} k$,
			which implies $j \prec_{l(t_\alpha)} k$.
			Thus we have constructed a contradiction to global dominance $j \prec_{g[t_o, t_e]} k$.
		\end{enumerate}
	\end{case}

	\begin{case}[Partial dominance] \\
	Job $j$ and job $k$ have no global precedence in time interval $[t_o, t_e]$. In this case there exists $t^*_{jk} \in (t_o, t_e)$.

		\begin{enumerate}

			\item $t^*_{jk} < t^*_{\alpha j}$: 
			On one hand by the definition of $\alpha$ we have $\varphi_k(t_o) > \varphi_j(t_o)$.
			On the other hand for the start time of job $k$, $t_k > t_\alpha > t^*_{\alpha j} + p_j > t^*_{jk} + p_j$.
			According to Lemma \ref{lemma: left_or_right}
			job $j$ precedes job $k$ does not violate the dominance properties between $i, j$.
			Figure \ref{fig: case_1} depicts the auxiliary functions of job $\alpha, j, k$ of this scenario.

			\item $t^*_{jk} = t^*_{\alpha j}$: 
			We claim this scenario is impossible.
			If $t^*_{jk} = t^*_{\alpha j}$ we will have $t^*_{jk} = t^*_{\alpha j} = t^*_{\alpha k}$.
			By the definition of $i, j$ this relation will lead to $t^*_{\alpha k} < t^*_{\alpha j} + p_j < t_{\alpha} < t^*_{\alpha k} + p_k$, that is, $t_\alpha \in [t^*_{\alpha k}, t^*_{\alpha k} + p_k)$.
			Since $t_\alpha \notin B_\alpha$ and $[t^*_{\alpha k}, t^*_{\alpha k} + p_k) \subseteq B_\alpha$,
			we have reached a contradiction. See Figure \ref{fig: case_3} for the auxiliary functions of job $\alpha, j, k$ of this scenario

			\item $t^*_{jk} > t^*_{\alpha j}$: At first we show $t^*_{jk} > t^*_{\alpha k}$.
			Assume by contradiction that $t^*_{\alpha j} < t^*_{jk} < t^*_{\alpha k}$,
			which would imply $\alpha \prec_{g(t^*_{\alpha j}, t_e)} j$ and $\alpha \prec_{g(t_o, t^*_{\alpha k})} k$. 
			Then for any time $t \in (t^*_{\alpha j}, t^*_{\alpha k})$ we have the relation $j \prec_{l(t^*_{\alpha j})} \alpha \prec_{l(t^*_{\alpha j})} k$.
			However, this is impossible since $t^*_{jk} \in (t^*_{\alpha j}, t^*_{\alpha k})$. 
			Therefore, we have $t^*_{jk} > t^*_{\alpha k}$,
			and job $j$ starts before time $t^*_{kj}$.
			In that case, having job $j$ precedes job $k$ does not violate the dominance properties between $i, j$,
			our claim holds.
			See Figure \ref{fig: case_2} for this scenario.
		\end{enumerate}		
	\end{case}
	\qed
\end{proof}

\begin{figure}

	\begin{minipage}[b]{.3\textwidth}
		\centering
			\begin{tikzpicture}[domain=-2:5,yscale=.7, xscale=.7, samples=200,>=latex,thick]
			  \draw[thick,->] (-0.5,0.5) -- (3.7,0.5) node[right] {\scriptsize $t$};
			  \draw[thick,->] (0,0) -- (0,4.7) node[left] {\scriptsize $\varphi(t)$};
			  \draw (0,0.5) node[below left] {\scriptsize $t_o$};
			  \coordinate (O) at (0,0);

			  \draw[name path=job_alpha, domain=0:3.7, color=black] plot (\x, {(3.95)/(0.98*(0.98 + \x)}); 
			  \draw (-0.1, 4) node[left] {\scriptsize $\varphi_\alpha(t)$};

			  \draw[name path=job_j, domain=-0:3.7, color=black] plot (\x, {(15)/(2.5*(2.5+ \x)});  
			  \draw (-0.1, 2.5) node[left] {\scriptsize $\varphi_k(t)$};

			  \draw[name path=job_k, domain=-0:3.7, color=black] plot (\x, {(160)/(9*(9 + \x)});
			  \draw (-0.1, 1.95) node[left] {\scriptsize $\varphi_j(t)$};

			  \draw [dashed, thin] (1.3714309278350516, 1.7141104155710498) -- (1.3714309278350516, 0.5) node[below] {\scriptsize $t^*_{\alpha j}$};

			  \draw [dashed, thin] (2.1308808290155445, 1.2956498388829214) -- (2.1308808290155445, 0.5) node[below] {\scriptsize $t^*_{\alpha k}$};

			  \draw [dashed, thin] (0.8113207547169812, 1.8119658119658117) -- (0.8113207547169812, 0.5) node[below] {\scriptsize $t^*_{j k}$};
			\end{tikzpicture}
		\caption{$t^*_{jk} < t^*_{\alpha j}$.}
		\label{fig: case_1}	
	\end{minipage}
	\begin{minipage}[b]{.3\textwidth}
		\centering
			\begin{tikzpicture}[domain=2:5,yscale=.7, xscale=.7, samples=200,>=latex,thick]
			  \draw[thick,->] (-0.5,0.5) -- (3.7,0.5) node[right] {\scriptsize \scriptsize$t$};
			  \draw[thick,->] (0,0) -- (0,4.7) node[left] {\scriptsize $\varphi(t)$};
			  \draw (0 ,0.5) node[below left] {\scriptsize $t_o$};
			  \coordinate (O) at (0,0);

			  \draw[name path=job_alpha, domain=0:3.7, color=black] plot (\x, {(3.95)/(0.98*(0.98 + \x)}); 
			  \draw (-0.1, 4) node[left] {\scriptsize $\varphi_\alpha(t)$};

			  \draw[name path=job_j, domain=0:3.7, color=black] plot (\x, {(15)/(2.5*(2.5+ \x) - 1.05});  
			  \draw (-0.1, 2.95) node[left] {\scriptsize $\varphi_k(t)$};

			  \draw[name path=job_k, domain=0:3.7, color=black] plot (\x, {(160)/(9*(9 + \x)});
			  \draw (-0.1, 2) node[left] {\scriptsize $\varphi_j(t)$};

			  \draw [dashed, thin] (1.3714309278350516, 1.7141104155710498) -- (1.3714309278350516, 0.5) node[below] {\scriptsize $t^*_{\alpha j}$};

			\end{tikzpicture}
		\caption{$t^*_{jk} = t^*_{\alpha j}$.}
		\label{fig: case_3}	
	\end{minipage}	
	\begin{minipage}[b]{.3\textwidth}
		\centering
			\begin{tikzpicture}[domain=-2:5,yscale=.7, xscale=.7, samples=200,>=latex,thick]
			  \draw[thick,->] (-0.5,0.5) -- (3.7,0.5) node[right] {\scriptsize$t$};
			  \draw[thick,->] (0,0) -- (0,4.7) node[left] {\scriptsize$\varphi(t)$};
			  \draw (0,0.5) node[below left] {\scriptsize$t_o$};
			  \coordinate (O) at (0,0);

			  \draw[name path=job_alpha, domain=0:3.7, color=black] plot (\x, {(5.5)/(1.2*(1.2 + \x)}); 
			  \draw (-0.1, 1.85) node[left] {\scriptsize$\varphi_k(t)$};

			  \draw[name path=job_j, domain=0:3.7, color=black] plot (\x, {(15)/(2.2*(2.2+ \x)});  
			  \draw (-0.1, 3.82) node[left] {\scriptsize$\varphi_\alpha(t)$};

			  \draw[name path=job_k, domain=0:3.7, color=black] plot (\x, {(150)/(9*(9 + \x)});
			  \draw (-0.2, 3.1) node[left] {\scriptsize$\varphi_j(t)$};

			  \draw [dashed, thin] (0.85084745762712, 2.234848484848484) -- (0.85084745762712, 0.5) node[below] {\scriptsize$t^*_{\alpha j}$};

			  \draw [dashed, thin] (1.7586206896551724, 1.5491452991452992) -- (1.7586206896551724, 0.5) node[below] {\scriptsize$t^*_{\alpha k}$};

			  \draw [dashed, thin] (2.507692307692307, 1.448306595365419) -- (2.507692307692307, 0.5) node[below] {\scriptsize$t^*_{j k}$};
			\end{tikzpicture}
		\caption{$t^*_{jk} > t^*_{\alpha j}$.}
		\label{fig: case_2}	
	\end{minipage}

\end{figure}

Combining Lemma \ref{lemma: greedy} and Lemma \ref{lemma: independent},
we can identify a strong connection between Algorithm \ref{algo: fast_schedule}
and potential schedules.

\begin{lemma}
\label{lemma: least_one}
	Whenever Algorithm \ref{algo: fast_schedule} adds a new branch for an instance $(J, t)$, 
	there is at least one potential schedule on that branch.
\end{lemma}

\begin{proof}
	Suppose after we find a possible start time of job $\alpha$ 
	the remaining jobs are divided into $J_l$ and $J_r$.
	By Lemma \ref{lemma: greedy} both job sets have at least one partial potential schedule.
	Denote the partial schedule of $J_l$ by $S_l$ and the partial schedule of $J_r$ by $S_r$,
	then according to Lemma \ref{lemma: independent} the jointed schedule $[S_l, \alpha, S_r]$ is also a potential schedule.
	\qed
\end{proof} 

Now we are ready to prove Theorem \ref{theorem: num_sol}.

\begin{proof}[Proof of Theorem \ref{theorem: num_sol}]
	We only need to generate all $t^*_{ij}$ and sorted them once, 
	which takes $\bigO(n^2 \log n)$ time.
	While finding the possible start times of $\alpha$,
	with the already calculated $t^*_{ij}$,
	we can construct the set $C_J^\alpha$ and the mapping $M_{J}^{\alpha}$ in $\bigO(n)$ time.
	The iteration over each subinterval $[c_q, c_{q + 1})$ also needs at most $\bigO(n)$ time.
	Therefore, it takes at most $\bigO(n)$ time to find a potential start time of a job.

	By Lemma \ref{lemma: least_one} there exists at least one potential schedule with job $\alpha$ starting at that time. 
	For any potential schedule there are $n$ start times to be decided,
	as a consequence,
	Algorithm \ref{algo: fast_schedule} finds each potential schedule in at most $\bigO(n^2)$ time.
	To find all the potential schedules it will take $\bigO(n^2 (\log n + K))$ time.
	\qed
\end{proof}
%

%
%

\section{Experimental Study} 
\label{sec:experimental_study}

We code our algorithm with Python 3 and perform the experiment on a Linux machine with one Intel Core i7-9700K @ 3.6Ghz $\times$ 8 processor and 16Gb RAM.
Notice that our implementation only invokes one core at a time.

For experimental data set we adopt the method introduced by Höhn and Jacobs \cite{bader_experimental_2012} to generate random instances.
For an instance with $n$ jobs, 
the processing time $p_i$ of job $i$ is an integer generated from uniform distribution ranging from 1 to 100, 
whereas the priority weight $w_i = 2^{N(0, \sigma^2)} \cdot p_i$ with $N$ being normal distribution.
Therefore, a random instance is characterized by two parameters $(n, \sigma)$.
According to previous results \cite{bader_experimental_2012,vasquez_for_2015},
instances generated with smaller $\sigma$ are more likely to be harder to solve,
which means we can roughly tune the hardness of instances by adjusting the value $\sigma$.

At first we compare our algorithm with the $A^*$ algorithm given by Vásquez \cite{vasquez_for_2015}.
This algorithm modeles airplane refueling problem as a shortest path problem on a directed acyclic graph $G$.
The vertexes of $G$ consist of all subset of $J$ and arcs are linked according to Corollary \ref{rule: new_rule}，
and a path from vertex $\{\}$ to $J$ corresponds to a schedule.
Since there will be $2^n$ vertexes in $G$ and it takes $\bigO(n)$ time to process each vertex in the worst case, 
the complexity of this algorithm is $\bigO(n \cdot 2^n)$.
This part of experiment is conducted on data set $S_1$,
which is generated with $\sigma=0.1$ and job size ranges from $\{10, 20, \dots, 140\}$,
where for each configuration there are 50 instances.

Secondly, we evaluate the empirical performance of Algorithm \ref{algo: fast_schedule} 
on data set $S_2$.
In detail, for each job size $n$ in $\{100, 500, 1000, 2000, 3000 \}$ and for each $\sigma$ value $\{0.1, 0.101, 0.102, \dots, 1 \}$, there are 5 instances,
that is, $S_2$ has 4505 instances in total.

At last, we generate data set $S_3$ to examine the relations of instance hardness with the number of potential schedules and the value of $\sigma$.
This data set has 5 instances of 500 jobs for each $\sigma$ from $\{0.100, 0.101, 0.102, \dots, 1\}$.

\subsubsection{Comparison with $A^*$} 
\label{ssub:comparison_with_}
We consider the ratio between the average running time of 
$A^*$ and Algorithm \ref{algo: fast_schedule} on data set $S_1$.
As shown in Figure \ref{fig: compare},
our algorithm outperforms $A^*$ on all sizes
and the speed up is more significant on instances with larger size.

For instances with 140 jobs, our algorithm is over 100 times faster.

\begin{figure}[H]
	\centering
	\includegraphics[width=.6\textwidth]{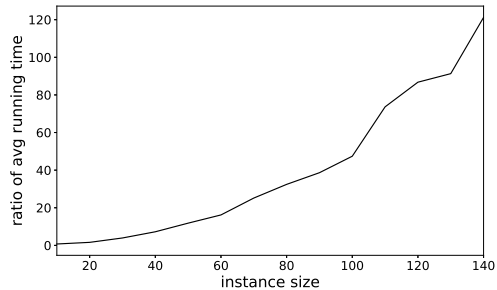}
	\caption{Speed up factor on dependence of instance size, on data set $S_1$.}
	\label{fig: compare}
\end{figure}

\subsubsection{Empirical performance} 
\label{ssub:empirical_performance}
Table \ref{tab: running_time} presents the running time of Algorithm \ref{algo: fast_schedule} on data set $S_2$.
We set a 10000 seconds timeout while performing the experiment.
When our algorithm does not solve all the instances within this time bound, 
we present the percentage of the solved instances.
For all instances with less than 2000 jobs 
as well as most instances with 2000 and 3000 jobs,
our algorithm returns the optimal solution within the time bound.
While for hard instances of 2000 and 3000 jobs generated with $\sigma$ less than $0.2$,
our algorithm can solve $92\%$ and $79\%$ of the instances within the time bound respectively.

\begin{table}[H]
\centering
\caption{Running time on data set $S_2$.}
\label{tab: running_time}
\scriptsize
\begin{tabularx}{\textwidth}{@{}p{.1\textwidth}<{\centering}YYYYYYYYYY@{}}
\toprule
\multirow{2}{*}{$|J|$} & \multicolumn{2}{c}{100}  & \multicolumn{2}{c}{500} 
                       & \multicolumn{2}{c}{1000} & \multicolumn{2}{c}{2000} 
                       & \multicolumn{2}{c}{3000} \\ 
				       \cmidrule(lr){2-3} \cmidrule(lr){4-5} \cmidrule(lr){6-7} \cmidrule(lr){8-9} \cmidrule(lr){10-11}        
                       & $\mathrm{avg.}$  & $\mathrm{std.}$  & $\mathrm{avg.}$  & $\mathrm{std.}$  & $\mathrm{avg.}$  & $\mathrm{std.}$        
                       & $\mathrm{avg.}$  & $\mathrm{std.}$  & $\mathrm{avg.}$  & $\mathrm{std.}$  \\ 
\midrule
$[0.1, 0.2)$  & 0.37   & 0.16   & 62.26   & 128    & 645.4   & 2223    & 92.40\% & -       & 79.60\%  & -      \\
$[0.2, 0.3)$  & 0.30   & 0.05   & 13.38   & 8.5    & 65.08   & 43.23   & 328.6   & 210.2   & 955.4    & 871.5  \\
$[0.3, 0.4)$  & 0.27   & 0.03   & 7.11    & 9.27   & 41.91   & 12.46   & 222.7   & 128.0   & 538.6    & 339.4  \\
$[0.4, 0.5)$  & 0.26   & 0.02   & 8.16    & 0.91   & 35.74   & 5.00    & 172.8   & 38.70   & 421.5    & 87.15  \\
$[0.5, 0.6)$  & 0.25   & 0.01   & 7.64    & 0.59   & 33.27   & 3.24    & 159.9   & 28.54   & 392.0    & 53.67  \\
$[0.6, 0.7)$  & 0.25   & 0.01   & 7.37    & 0.38   & 31.84   & 2.74    & 151.0   & 14.37   & 367.3    & 41.37  \\
$[0.7, 0.8)$  & 0.25   & 0.01   & 7.18    & 0.35   & 30.52   & 1.52    & 142.1   & 12.97   & 349.4    & 22.09  \\
$[0.8, 0.9)$  & 0.24   & 0.01   & 7.07    & 0.32   & 30.10   & 1.62    & 144.8   & 9.07    & 341.4    & 22.79  \\ 
$[0.9, 1.0]$  & 0.24   & 0.01   & 6.94    & 0.23   & 29.28   & 1.00    & 141.5   & 7.41    & 327.3    & 16.93  \\ 
\bottomrule
\end{tabularx}
\end{table}

\subsubsection{Instance hardness} 
\label{ssub:instance_hardness}
Figure \ref{fig: hardness} depicts relations of different hardness indicators.
The chart on the left shows the number of potential schedules and 
the running time of Algorithm \ref{algo: fast_schedule} on data set $S_3$.
Since both variables cover a large range,
we present the figure as a log-log plot.
As indicated by Theorem \ref{theorem: num_sol},
there is a strong correlation between the number of potential schedules and the running time of Algorithm \ref{algo: fast_schedule}.
Therefore, the number of potential schedules can serve as a rough measure of the instance hardness.
For the relation between the number of potential schedules and the $\sigma$ value, 
the second chart exhibits that 
instances generated with small $\sigma$ are more likely to have a large number of potential schedules.
The relation above is more obvious on smaller $\sigma$,
while on larger $\sigma$ the difference between the number of potential schedules is less significant.

\begin{figure}[H]

	\begin{minipage}[b]{.5\textwidth}
		\centering
		\includegraphics[width=\textwidth]{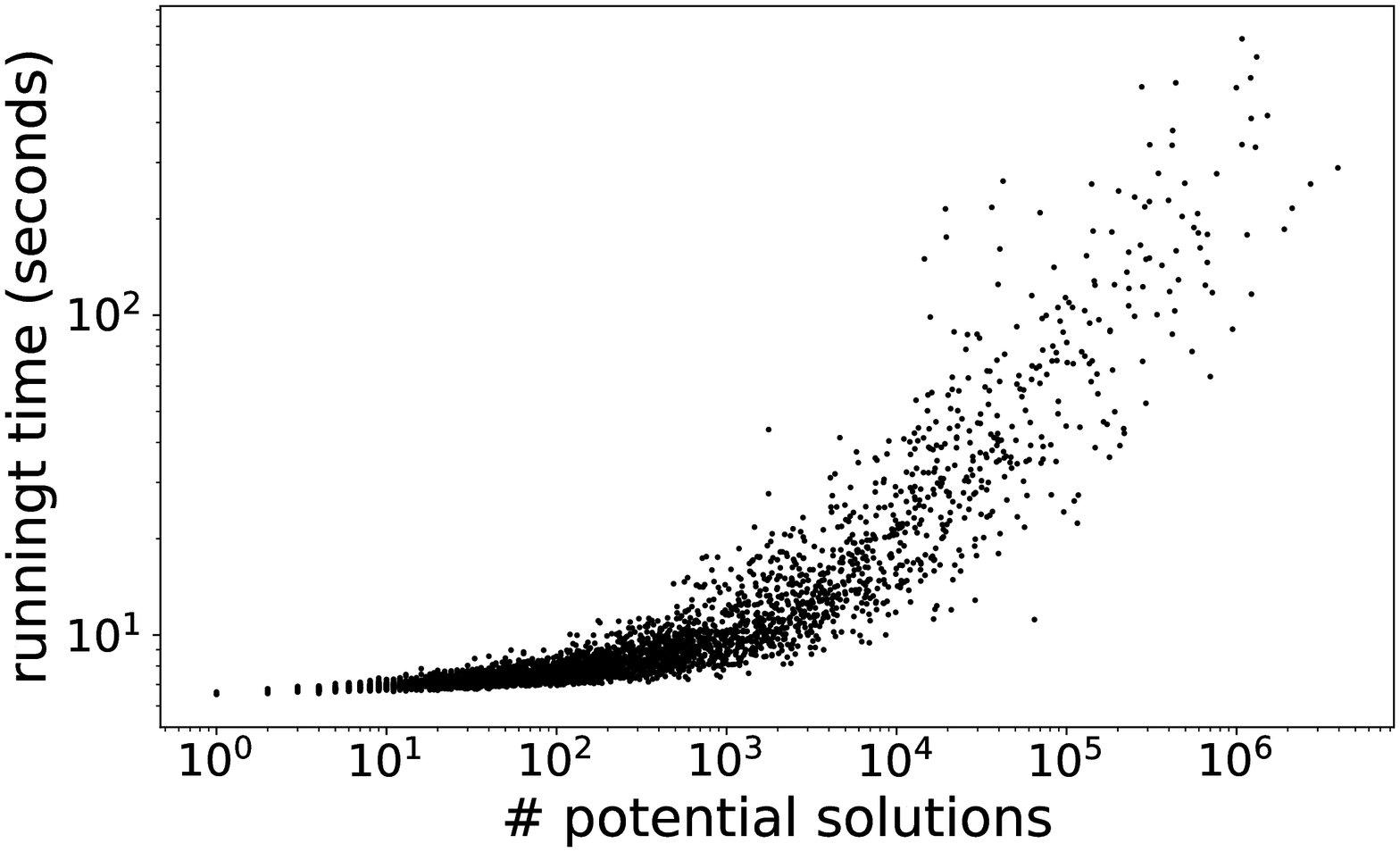}
	\end{minipage}
	\begin{minipage}[b]{.5\textwidth}
		\centering
		\includegraphics[width=\textwidth]{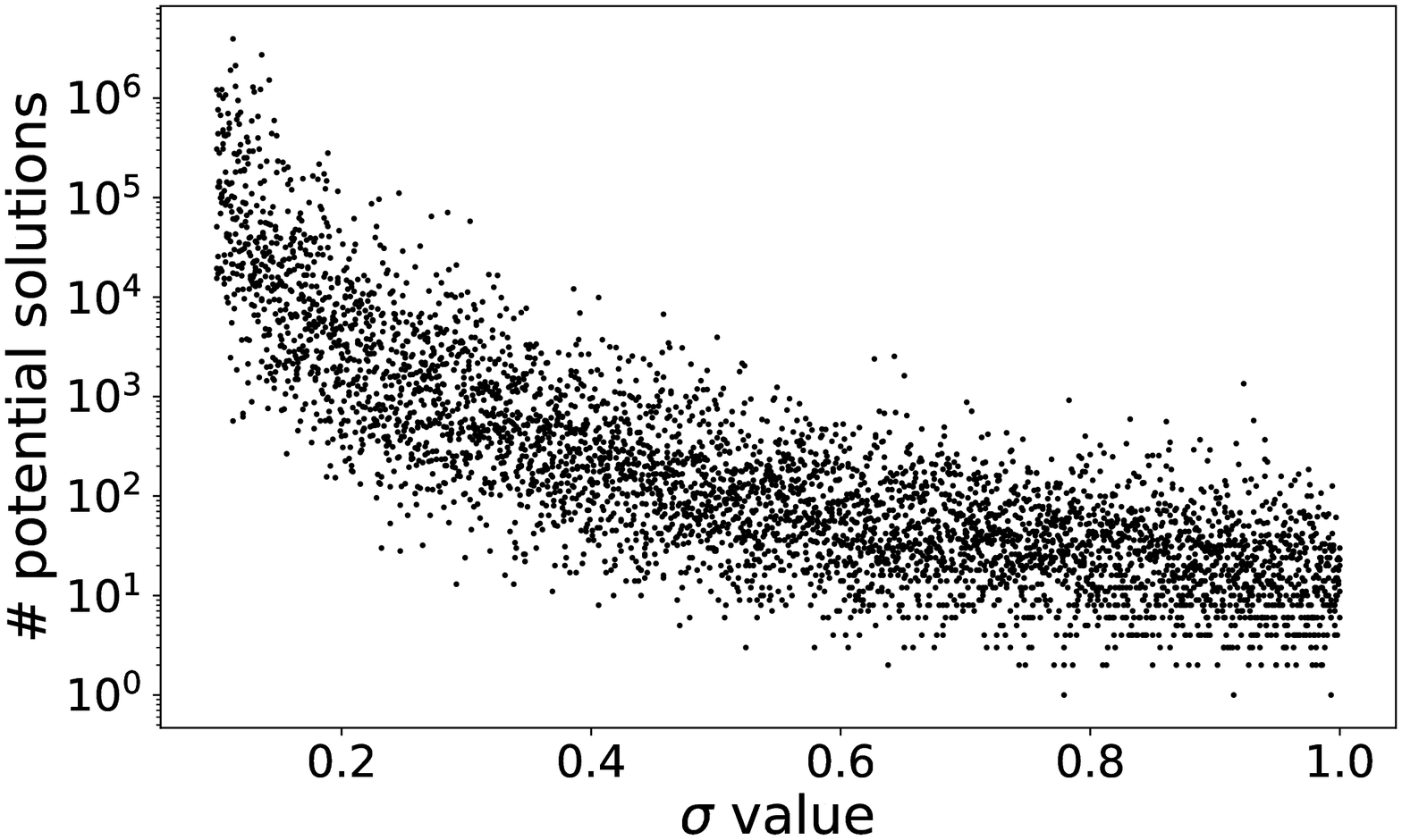}
	\end{minipage}
	\caption{Different indicators of instance hardness, on data set $S_3$.}
	\label{fig: hardness}
\end{figure}

\section{Conclusions} 
\label{sec:conclusions}
We devise an efficient exact algorithm for airplane refueling problem.
Based on the dominance properties of the problem, 
we propose a method that can prefix some jobs' start times and determine the relative orders among jobs in a potential schedule.
This technique enables us to solve airplane refueling problem in a recursive manner.
Our algorithm outperforms the state of the art exact algorithm on random generated data sets.
For large instances with hard configurations that can not be tackled by previous algorithms,
our algorithm can solve most of them in a reasonable time.

The empirical efficiency of our algorithm can be attributed to two factors.
First, 
our algorithm explores only the branches that contain potential schedules.
Second,
on the root node of each branch,
the problem is further divided into smaller subproblems,
which can also speedup the searching process.
Theoretically, 
we prove that the running time of our algorithm is upper bounded by the number of potential schedules
times a polynomial overhead in the worst case.

Another contribution of this work is that we give some new structural properties of airplane refueling problem, 
which may be helpful in understanding the computational complexity of the problem.

%
%

\bibliography{NVehicle}

\end{document}